\documentclass[
aps,pra,
reprint,
superscriptaddress,
amsmath,
amssymb
]{revtex4-1}

\usepackage[utf8]{inputenc}
\usepackage{times}
\usepackage{microtype}
\usepackage{graphicx}
\usepackage{upgreek}
\usepackage{bm}
\usepackage[
    breaklinks,
    pdftex,
    colorlinks=true,
    linkcolor=teal,
    citecolor=teal,
    urlcolor=teal
]{hyperref}

\usepackage{xcolor}
\usepackage[normalem]{ulem}
\usepackage{placeins}
\usepackage{amsmath}
\usepackage{amsthm}
\usepackage{dcolumn}
\usepackage{comment}

\usepackage{filecontents}
\usepackage[T1]{fontenc}
\usepackage{csquotes}
\usepackage[german,ngerman,english]{babel}
\usepackage{centernot}
\usepackage{mathtools}
\usepackage{ifthen}
\usepackage{tikz}
\usetikzlibrary{calc}
\usetikzlibrary{decorations.pathreplacing}
\usetikzlibrary{decorations}
\usetikzlibrary{positioning}
\usepackage{pgfplots}
\usepackage{xxcolor}
\definecolor{structure}{rgb}{0.23,0.4,0.7}
\usepackage[normalem]{ulem}

\newtheorem{lemma}{Lemma}
\newtheorem{theorem}{Theorem}

\newtheorem{definition}{Definition}

\newsavebox{\blocksavebox}

\definecolor{niceblue}{rgb}{0.33,0.5,0.8}

\newcommand{\rr}{\mathbb{R}}
\newcommand{\cc}{\mathbb{C}}

\newcommand{\refsub}[2]{\hyperref[#1]{\ref*{#1}#2}}

\newcommand{\coloneqq}{\mathrel{\vcentcolon\mkern-1.2mu=}} 

\newcommand{\norm}[2][]{
  \ifthenelse{\equal{#1}{}}
    {\left\| {#2} \right\|}
    {\ifthenelse{\equal{#1}{uinv}}
      {\left\vert\kern-0.25ex\left\vert\kern-0.25ex\left\vert {#2} \right\vert\kern-0.25ex\right\vert\kern-0.25ex\right\vert}
      {\left\| {#2} \right\|_{#1}}
    }
}

\newcommand{\taverage}[2][]{
  \ifthenelse{\equal{#1}{}}
  {\overline{#2}}
  {\overline{#2}^{#1}}
}

\newcommand{\tracedistance}[3][]{
  \ifthenelse{\equal{#2}{}}
  {\ifthenelse{\equal{#3}{}}
    {\mathcal{D}_{#1}}{}
  }{
    \ifthenelse{\equal{#1}{}}
    {\mathchoice{\operatorname{\mathcal{D}}\left(#2,#3\right)}{\operatorname{\mathcal{D}}(#2,#3)}{\operatorname{\mathcal{D}}(#2,#3)}{\operatorname{\mathcal{D}}(#2,#3)}}
    {\mathchoice{\operatorname{\mathcal{D}}_{#1}\left(#2,#3\right)}{\operatorname{\mathcal{D}}_{#1}(#2,#3)}{\operatorname{\mathcal{D}}_{#1}(#2,#3)}{\operatorname{\mathcal{D}}_{#1}(#2,#3)}}
  }
}
\newcommand{\fidelity}[3][]{
  \ifthenelse{\equal{#2}{}}
  {\ifthenelse{\equal{#3}{}}
    {\mathcal{F}_{#1}}{}
  }{
    \ifthenelse{\equal{#1}{}}
    {\mathchoice{\operatorname{\mathcal{F}}\left(#2,#3\right)}{\operatorname{\mathcal{F}}(#2,#3)}{\operatorname{\mathcal{F}}(#2,#3)}{\operatorname{\mathcal{F}}(#2,#3)}}
    {\mathchoice{\operatorname{\mathcal{F}}_{#1}\left(#2,#3\right)}{\operatorname{\mathcal{F}}_{#1}(#2,#3)}{\operatorname{\mathcal{F}}_{#1}(#2,#3)}{\operatorname{\mathcal{F}}_{#1}(#2,#3)}}
  }
}

\newcommand{\Sr}[3][]{
  \ifthenelse{\equal{#1}{}}
    {\operatorname{\mathnormal{S}}(#2\|#3)}
    {\operatorname{\mathnormal{S}}_{#1}(#2\|#3)}
}

\DeclareMathOperator{\1}{\mathbb{I}}



\newcommand{\tr}{{\rm tr}}

\definecolor{jens}{rgb}{0.1,0.5,0.1}
\definecolor{martin}{rgb}{0,0,1.0}

\newcommand{\beq}[0]{\begin{equation}}
\newcommand{\eeq}[0]{\end{equation}}

\newcommand{\hide}[1]{}

\begin{document}

\title{A compellingly simple proof of the speed of sound for interacting bosons}

\author{J.\ Eisert}
\address{Dahlem Center for Complex Quantum Systems, Freie Universit{\"a}t Berlin, 14195 Berlin, Germany}
\address{Helmholtz-Zentrum Berlin f{\"u}r Materialien und Energie, 14109 Berlin, Germany}

\begin{abstract}
On physical grounds, one expects locally interacting quantum many-body systems to feature a finite group velocity. This intuition is rigorously underpinned by Lieb-Robinson bounds that state that locally interacting Hamiltonians with finite-dimensional constituents on suitably regular lattices always exhibit such a finite group velocity. This also implies that causality is always respected by the dynamics of quantum lattice models. It had been a long-standing open question whether interacting bosonic systems also feature finite speeds of sound in information and particle propagation, which was only recently resolved. This work proves a strikingly simple such bound for particle propagation---shown in literally a few elementary, yet not straightforward, lines---for generalized Bose-Hubbard models defined on general lattices, proving that appropriately locally perturbed stationary states feature a finite speed of sound in particle numbers.
\end{abstract}

\maketitle

The speed of sound determines how quickly information can travel through a medium like air. Similarly, 
quantum lattice models in condensed matter physics are expected to exhibit a speed of sound, with excitations propagating ballistically through the system. This is often taken as a given, and the idea that quantum lattice models with finite-range interactions should feature a finite speed of sound is encapsulated in the famous Lieb-Robinson bounds \cite{liebrobinson,Hastings-CMP-2006,Nachtergaele-CMP-2006,SpeedLimits}. For spin models,
these bounds are typically expressed in terms of operator norms of commutators, where one observable is evolved in time within the Heisenberg picture. While the physical interpretation of these expressions may not be immediately clear, various readings of finite information propagation and transport in non-equilibrium scenarios can be derived from them \cite{1408.5148, PolkovnikovReview,christian_review}. Lieb-Robinson bounds, in this context, establish the maximum speed at which information or correlations can propagate through a lattice. In any local lattice model with finite-dimensional constituents, information is largely confined to a sound cone, with propagation outside this cone exponentially suppressed.


Indeed, there is a wealth of further physically meaningful properties that can be derived from them. 
For instance, gapped phases of matter always exhibit exponentially decaying correlations \cite{Hastings-CMP-2006, Nachtergaele-CMP-2006}, which can be proven using these bounds. They help approximate ground states of gapped quantum spin systems \cite{0904.4642}, provide insights into the stability of quantum phases \cite{1109.1588, NachtergaelePhases}, and support the area law for entanglement entropy \cite{Marien2016, AreaReview}, a key feature of quantum many-body systems. Additionally, Lieb-Robinson bounds are used to demonstrate the stability of topological order under local perturbations \cite{1001.0344}. Notably, they play a central role in proving the quantization of Hall conductance for interacting electrons on a torus \cite{HallConductance}.

Even though originally devised for local Hamiltonian systems with finite-dimensional
constituents, they have been generalized to be applicable for long-ranged interactions
 \cite{PhysRevLett.113.030602}, a situation in  which still a sound cone can be identified if the interactions are sufficiently rapidly
 decaying \cite{LongRangedCone} (even though this will break down if the interactions become too long-ranged 
 \cite{1309.2308}). Pushing this mindset even further, analogues for open quantum systems
 have been formulated  \cite{1111.4210,LRreviewchapter}, including ones on long-ranged 
 open systems \cite{OpenLR}. Overall, it seems fair to say that Lieb-Robinson bounds have
 been established as one of the pillars on which the mathematical physics of quantum lattice
 models rests, with substantial and important  implications for condensed matter physics.

A long-standing question had remained unresolved throughout the developments mentioned above over many years, 
despite considerable efforts, and has only recently been settled: This is the question of whether a finite speed of sound could also be
identified for a ubiquitous family of physical systems: These are \emph{bosonic lattice systems}. The core issue lies in the fact that, in these systems, the local dimension is no longer finite, making the previously applicable proofs inapplicable. Finding strategies to address this puzzle is more complex than it may initially seem, as the assumption of finite local dimension is far from a mere technicality. It has become clear that for 
non-interacting bosonic systems, a finite group velocity can still be identified \cite{0803.0890,Anharmonic}, leading to strong bounds for equilibration 
\cite{AnalyticalQuench,GluzaEisertFarrelly} in non-equilibrium situations \cite{1408.5148,PolkovnikovReview}. Early studies on interacting systems have used elegant proof techniques, but they have focused on physically rather implausible systems \cite{Anharmonic}. It is even possible to construct contrived local interacting models that violate causality bounds  \cite{EisertGross09}, making the question of sound speeds in interacting models even more intriguing.


What one is commonly rather interested in are interacting bosonic particle Hamiltonians,
most prominently the Bose-Hubbard model. The painful lack of insights into the situation of Bose-Hubbard
like systems over the years has been aggravated by the fact that already a number of experiments have been
performed with ultra-cold atoms that exhibit a ballistic particle propagation in Bose-Hubbard 
systems \cite{1111.0776,Expansion}. Only very recently,  this long-standing question has been resolved affirmatively: There is a finite speed 
of particle and information propagation for interacting particle Hamiltonians, following a spectacular surge of 
academic activity
 \cite{Saito,PhysRevX.12.021039,BosonicLightCone,PhysRevLett.128.150602}. 
This breakthrough has involved significant technical progress and intricate and lengthy proof techniques.

This work presents a compellingly simple proof of macroscopic particle propagation in interacting 
 bosonic particle Hamiltonians, including the famous Bose-Hubbard model, defined on possibly 
 exotic lattices: In fact, a finite \emph{genuine speed of sound} in particle numbers follows from literally
 a few elementary and physically motivated lines:  Particle-like local excitations will propagate
basically like $Vt$ in time $t\geq 0$ with a speed of sound $V>0$, up to spatially exponentially suppressed
corrections, for arbitrary translationally invariant initial stationary states. 
Again, the  proof is strikingly simple. Even with all details explained and the argument meeting rigorous standards, it remains brief. 
No logarithmic corrections to the sound cone arise, and there are no issues with the relevant bosonic operators being unbounded. 
This result is aimed at demystifying the concept of sound speed for bosons in particle propagation, 
and is expected to have a strong didactical value.

 \begin{figure}[tb]
\centering
\includegraphics[width=0.8\columnwidth]{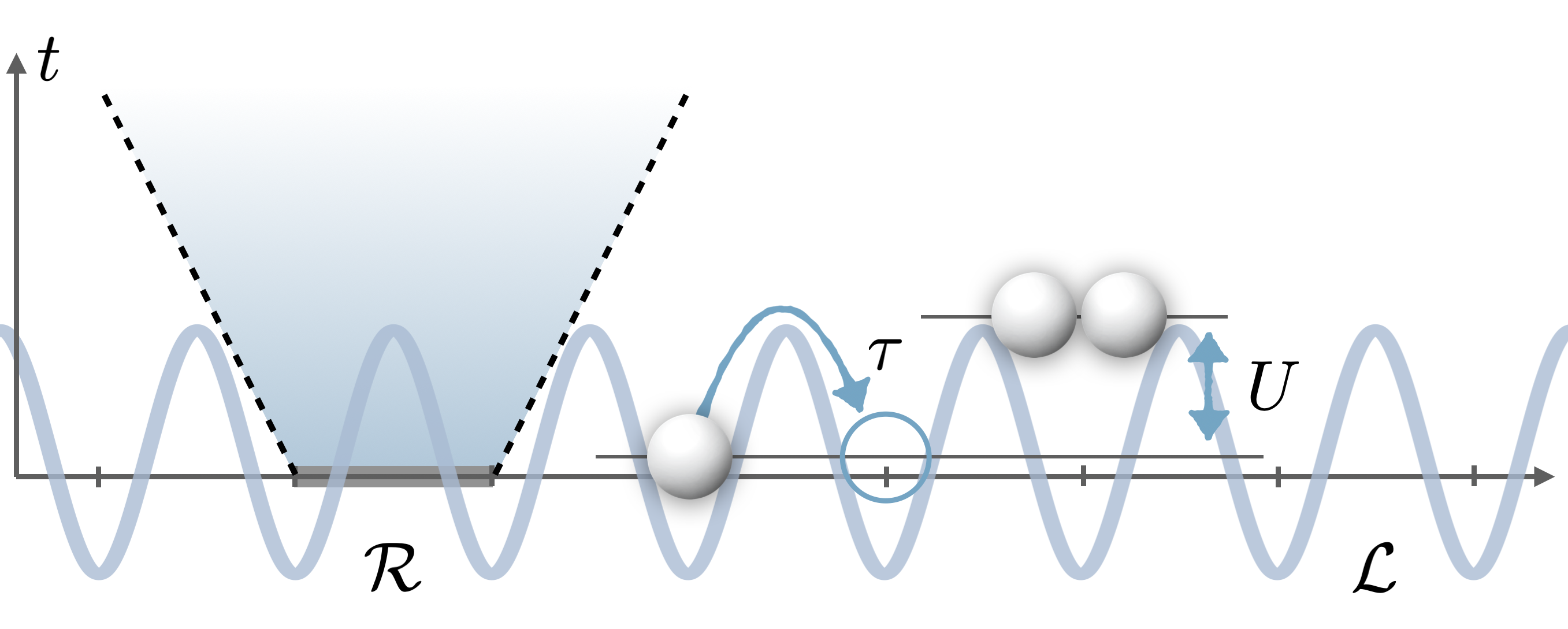}
\caption{A schematic picture of the Bose-Hubbard model in one spatial dimension, 
with a hopping at rate $\tau>0$ 
to neighbouring ones and on-site 
interactions with strength $U>0$, which is a paradigmatic example of the general 
interacting bosonic particle Hamiltonians considered
here. This work shows a finite speed of sound for particle densities for  general such
models.
Local excitations initially confined to a region ${\cal R}$ will at most ballistically
propagate through the lattice ${\cal L}$ with a speed of sound, 
up to exponentially small corrections.}
\label{Figure}
\end{figure}

\emph{Family of considered Hamiltonians.} Throughout this work, we consider
interacting bosonic particle Hamiltonians. This means that we have a graph $G=({\cal L},{\cal E})$ 
in mind with $|{\cal L}|=:n $ vertices associated with physical sites,
 on which quantum particles are being placed. Most prominently, this lattice would be a cubic one, 
 but we can allow for exotic lattices as well. Any such graph will feature the natural graph-theoretic distance ${\rm dist}(.,.)$. 
 For a one-dimensional quantum system, it is simply ${\rm dist}(j,k) =|j-k|$. $\langle j,k \rangle$ denotes 
 sites that are nearest neighbours for which ${\rm dist}(j,k)=1$. 
 Generally, the family of Hamiltonians we allow for are
of the form $\hat H :=  \hat T+ \hat V$ with 
\begin{equation}\label{Hamiltonian}
	\hat T := - \tau \sum_{\langle j,k \rangle}
	\hat b_j^\dagger \hat b_k ,\,
	\hat V:= \sum_{j=1}^n \hat A_j,
\end{equation}
with $\tau>0$, where $\{\hat A_j\}$ are operators supported on the site labelled $j$
that satisfy $[\hat n_j , \hat A_j]=0$ for all
$j=1,\dots, n$, so commute with the local number operator. This can even be an on-site high power of the number operator. 
This includes the familiar \emph{Bose-Hubbard Hamiltonian}  with
\begin{eqnarray}
	\hat V =  \frac{U}{2} \sum_{j=1}^n  \hat n_j (\hat n_j-1),
\end{eqnarray}
where  $\tau>0$ is the hopping strength reflecting tunnel processes from one site to a 
neighbouring one, $U>0$ denotes the repulsive interactions on a respective site, 
and $\mu>0$ is the chemical potential. Bosonic annihilation operators supported on site $j$
are denoted as $\hat b_j$, with $\hat n_j:= \hat b_j^\dagger \hat b_j$ being the local particle number operators.
For any such model and an operator $M$ defined on the 
 entire system, we
 can define its \emph{covariance matrix} $C(M)\in \cc^{n\times n}$ via its entries for $j,k=1,\dots n$ as
 \begin{equation}
 	C_{j,k}(M):= {\rm tr}( \hat b_j^\dagger \hat b_k M).
 \end{equation}
 
\emph{Causality bounds from perturbing stationary states.} We will consider the natural situation in which
the system is initially in a translationally invariant stationary state $\omega$. This stationary state would most naturally be 
the ground state or a thermal state 
$\omega=e^{-\beta \hat H}/{\rm tr} (e^{-\beta \hat H}) $
at some inverse temperature $\beta>0$,
but could also be any other stationary state with $[\omega, \hat H]=0$. At time $t=0$, this
state is locally excited by the addition of bosons. We show the ballistic propagation of the impact of this
excitation 
\begin{equation}
t\mapsto \rho(t) = e^{-it\hat H}
\Phi (\omega)
e^{it\hat H}
\end{equation}
as the state evolves in time $t\geq 0$ 
under the dynamics generated by $\hat H$.  In order to be as concrete and simple as possible, we consider the following  family of initial excitations.
It implies that Kraus operators transform ``covariantly'', and natural processes of boson addition satisfy this.

 \begin{definition}[Initial excitations]\label{IE} 
  We consider initial states of the form $\rho(0) =\Phi(\omega)$, where the channel $\Phi$ acts non-trivially
  on ${\cal R}$ only, and the Kraus operators of the quantum channel 
  $\omega\mapsto \sum_l \hat K_l \omega \hat K_l^\dagger$
  satisfy
  $[\hat n_j, \hat K_l] = \hat K_l \hat D_{l,j}$, where $[\hat D_{l,j},\hat n_k]=0$ and $\hat D_{l,j}\geq 0$ 
  for all $l$ and all $j,k=1,\dots, n$.
 \end{definition}
 
\emph{Approach of consistent particle numbers.}
Once again, the argument presented is simple, and in what follows, every important step of the argument
will be laid out in detail. At the same time,
in order to arrive at the approach followed, a number of prejudices 
need to be overcome, which will be commented upon in each relevant step.
Conceptually speaking---and also technically---the present argument builds closely upon that of 
Ref.\ \cite{1010.4576}, but here modified in a crucial manner. There, a
vector of particle numbers $\alpha\in \rr^n$ with
$\alpha_j(t) := {\rm tr}(\hat n_j (\rho(t)))$ is introduced and possible changes 
of this vector in time considered that are consistent with the  dynamics
generated by an interacting bosonic Hamiltonian. This will no longer directly work, since $\rho(t)-\omega$ is Hermitian, but not 
a positive semi-definite operator, leading to a breakdown of the proof techniques of Ref.\ \cite{1010.4576} resorting to the
Cauchy Schwarz inequality. Important is also the fact that the initial state is there locally uniquely
defined by a vanishing particle number, an insight that is implicitly used in the argument. This substantial obstacle is overcome here.

\emph{Relative particle numbers.} The findings are expressed in terms of the
\emph{vector of relative particle numbers} $\mathbf{x}\in \rr^n$,
\begin{equation}
x_j(t) := {\rm tr}(\hat n_j (\rho(t)-\omega)).
\end{equation}
This quantity has a crisp operational interpretation,
in that it is the literate deviation in particle number from that of the stationary state,
and the main diagonal of $C(\rho(t)) - C(\omega)$.
This quantity naturally  reflects particle-like excitations and reflects particle propagation.

\emph{Main result.} We are now in the position to state the main result. This statement is
deliberately formulated in general terms. It should be clear, however, that it precisely
states what is intuitively expected: If one locally perturbs a stationary state of
an interacting bosonic system in an arbitrary manner, for any lattice, the
impact of this perturbation will propagate through the lattice with a proper,
well-defined speed of sound, without logarithmic corrections, so that the impact outside the sound cone
is exponentially suppressed with the distance.

\begin{theorem}[Finite speed of sound of particle propagation for interacting bosons] 
For any interacting bosonic Hamiltonian (\ref{Hamiltonian})
defined on a lattice $G$ and any initial state $\rho(0)$ that locally perturbs
a stationary state $\omega$ as in Definition 1,  
adding particles locally on a region ${\cal R}\subset {\cal L}$, 
it holds true that
\begin{equation}
|x_j(t)| \leq c N_0 e^{vt-l}
\end{equation}
for all times $t\geq 0$, 
for all $j$ with ${\rm dist} (j,{\cal R})\geq l$,
where $c>0$ is a universal constant and $v: = v_0 + D$ is an upper bound to the speed of sound.
$D$ is the maximal vertex degree of the lattice $G$,
$v_0 := \chi \Delta \tau$, where $\chi>0$ is an absolute constant
and $\Delta:= \|M\|/2$, $M$ being the adjacency matrix of $G$.
The constant $N_0>0$ depends on the initial state and $\Phi$ only.
\end{theorem}
It is important to note that $v>0$ is indeed an upper bound to a genuine speed of sound $V$, 
independent
of the system size. The constant is actually given by $C := 2\chi^2 /(\chi-1)\approx 9.95411$, where 
$\chi$ is the solution of $\chi \ln \chi = \chi+1$. 

Key to the argument is the following 
elementary---but not trivial---lemma. The subsequent statement may  
be surprising, as actually $\rho(t)-\omega\not\geq 0$, but
one can still conclude that $ C(\rho(t)) $ is larger than the covariance matrix $C(\omega)$ of the stationary state in matrix ordering.
While the argument is not particularly technical, it is not obvious either, and it is the particular way in which Eq.\ (\ref{premise}) is set up that renders the proof simple. So \emph{here} begins the short argument.
  
 \begin{lemma}[Positivity] For all times $t\geq 0$,
 \begin{equation}
 C(\rho(t)) \geq C(\omega).
 \end{equation}
 \end{lemma}
 
 \begin{proof} We have $C(\rho(0))\geq C(\omega)$ by assumption of the initial condition, reflecting the 
 addition of bosonic particles within ${\cal R}$. The interesting insight is that this matrix ordering is
 preserved over time. In order to do so, we will make use of the particle number as a reference operator, and
 we will now prove that
\begin{equation}\label{premise}
\sum_{j,k=1}^n \sum_l M_{j,k}  \hat K^\dagger_l  e^{it\hat H }\hat b_j^\dagger \hat b_k e^{-it\hat H }  \hat K_l    - 
\sum_{j=1}^n     \hat b_j^\dagger \hat b_j  \geq 0
 \end{equation}
for all $t\in [0,\infty)$ and 
for all $\mathbb{R}^{n\times n} \ni M\geq \mathbb{I}$ (in a ``continuous induction argument''). The set $S$ of times within $[0,\infty)$
for which Eq.\ (\ref{premise}) holds is not empty, as for $t=0$, this follows from 
the commutation relations in Definition 1. The set $S\backslash \{0\}$ is also open, in that $t\in S\Rightarrow \exists \varepsilon>0 : (t-\varepsilon, t+ \varepsilon)\cap [0,\infty)\subset S$.
What is more, the  set $S$ is  closed, in that if $t_n\in S, t_n\rightarrow t\Rightarrow t \in S$.
This allows us  to argue that Eq.\ (\ref{premise}) holds 
 true for all $t\in [0,\infty)$ for all $M\geq \mathbb{I}$.
 To do so, we treat the interacting and the hopping part separately, 
 invoking the Lie product formula, $e^{-i\varepsilon \hat H} \psi =
 e^{-i\varepsilon (\hat T+ \hat V)} \psi= 
  \lim_{m\rightarrow \infty} (e^{-i\varepsilon \hat T/m}+  e^{-i\varepsilon \hat V/m})^m \psi$
  for every Hilbert space vector $\psi$, with convergence in Hilbert space norm for each $\psi$, i.e., 
  strong operator convergence. The functional analysis of the Bose-Hubbard type models is well understood.
  The hopping part actually immediately 
  preserves the validity of
  Eq.\ (\ref{premise}), as time evolution by a time $\delta \coloneqq \varepsilon/m$
  is reflected by a map $M\mapsto O M O^T $ with $O\in O(n)$, so that  $O M O^T\geq \mathbb{I}$ is manifestly true
  if $M\geq \mathbb{I}$.
  The interacting part is slightly more involved. Making use of 
Definition 1, Eq.\ (\ref{premise}) becomes for time
$t+\delta$ 
     \begin{eqnarray}
     \label{complicated}
&&\sum_{j,k=1}^n M_{j,k} \sum_l \hat K^\dagger_l e^{i\delta \hat V }  e^{it\hat H }\hat b_j^\dagger \hat b_k e^{-it\hat H } e^{-i\delta \hat V }  \hat K_l  \\
&=&\sum_{j,k=1}^n M_{j,k} \sum_l \hat K^\dagger_l  e^{it\hat H }\hat b_j^\dagger \hat b_k e^{-it\hat H }  \hat K_l \nonumber \\
&+& i \delta \sum_{j,k=1}^n M_{j,k} \sum_l [\hat D  + \hat V, \hat K^\dagger_l  e^{it\hat H }\hat b_j^\dagger \hat b_k e^{-it\hat H }  \hat K_l] 
+O(\delta^2) .
\nonumber
 \end{eqnarray}
 Invoking now Eq.\ (\ref{premise}) and making use of the second property of the initial condition in Definition 1, 
one can argue that (\ref{complicated}) is larger than $\sum_{j=1}^n \hat n_j$ in matrix ordering
for an interval $(t-\delta,t+\delta)$. Taken together
with the above Lie product formula, it shows that Eq.\ (\ref{premise}) is also true for $ (t-\varepsilon, t+ \varepsilon)\cap [0,\infty)\subset S$
for all $M\geq \mathbb{I}$. Closedness is immediate. Applying Eq.\ (\ref{premise}) now to the reference state $\omega$ and taking the
trace gives rise to 
\begin{equation}
\sum_{j,k=1}^n M_{j,k} C_{j,k}(\rho(t)) \geq \tr( C(\omega))
\end{equation}
for all $t\in[0,\infty)$ and all $M\geq \mathbb{I}$, from which the 
 statement to be shown follows.
 \end{proof}

 Equipped with this insight, the argument of Theorem 1 becomes immediate and exactly follows that of Ref.\ \cite{1010.4576}, 
 which is still given for completeness.
 \begin{proof} For the vector of relative particle numbers, we find 
 \begin{equation}
 \left| \frac{d}{dt}  x_j(t) \right|  \leq  2\tau  \sum_{\langle j,k\rangle}  \tr \left( \hat b_j^\dagger \hat b_k (\rho(t) -\omega) \right)
 \leq 2\tau  \sum_{\langle j,k\rangle}  \left( x_j(t) x_k(t)\right)^{1/2}
 \end{equation}
 for $t\geq 0$. While $\rho(t) -\omega\not\geq 0$, this conclusion can still be drawn since $C(\rho(t) -\omega)\geq 0$. Linearizing this
gives
  \begin{equation}
\Bigl| \frac{d}{dt}  x_j(t) \Bigr|  \leq  \tau   
\Bigl(
 {\cal D} x_j(t) + \sum_{\langle j,k\rangle} x_k(t)
 \Bigr). 
 \end{equation}
An upper bound $\gamma_j(t)\geq x_j(t)$ is given by the solution of the differential equation
  \begin{equation}
 \Bigl| \frac{d}{dt}  \gamma_j(t) \Bigr|  = \tau   
 \Bigl(
 {\cal D} \gamma_j(t) + \sum_{\langle j,k\rangle} \gamma_k(t)
 \Bigr),
 \end{equation}
 which, upon exponentiation, gives the statement of the theorem.
 \end{proof}

 \emph{Summary.} This work has shown that one can derive a compellingly simple bound for the speed of sound of bosons for particle propagation, 
 an important
 foundational question that has long been unresolved. While the argument presented is fully rigorous, it is based on a very physical insight 
 that carries a lot of substance and information: After all, even in the case of large filling, one can pursue an argument of perturbed 
 covariance matrices. While the argument as such is not technically involved, and fits onto a lot less than a page, it is based on a 
 number of creative insights and is far from trivial. In particular, the topological argument that allows to conclude that the perturbed
 covariance matrix is larger than the unperturbed on in matrix ordering is not obvious. It is the hope that the present work provides
 physical substance to this old problem in quantum many-body physics and invites more and new sophisticated bounds on information
 propagation, and may assist making technical progress for technically demanding settings such as long-ranged Bose-Hubbard models
\cite{LemmLongRanged}. Given the clear physical intuition provided here, such progress seems plausible. 
  
 \emph{Acknowledgements.} While the core of this, then still incomplete, argument was announced in May 2021 and sent to some colleagues, 
 only now I have found the time to complete this work. In the meantime, as laid out in the introduction,
  I became aware of the beautiful body of recent works
\cite{Saito,PhysRevX.12.021039,BosonicLightCone,PhysRevLett.128.150602,LemmLongRanged},
that addresses a very similar question with different methods. Given the striking simplicity of the present bound
and the physical insight it gives rise to for particle propagation, 
I still decided to make this work publicly available, as it has added value,
as this work may well complement the bound presented there, it has didactical value,
and since it provides in some ways 
a stronger bound giving rise to a genuine sound cone without logarithmic corrections and holds
true for large classes of lattices. I would also like to thank 
N.\ Schuch, S.\ K.\ Harrison and 
T.\ J.\ Osborne for discussions at the time of the publication of Ref.\
\cite{1010.4576} and for M.\ Lemm for kind and highly valuable feedback on the manuscript. 
This work has been supported by the DFG (CRC 183) and the FQXi. It has also received funding from 
the Quantum Flagship (PASQuanS2), the BMFTR (MuniQCAtoms), and the European Research Council (DebuQC). 


%

\end{document}